\documentclass[letterpaper, 10 pt, conference]{ieeeconf}  

\IEEEoverridecommandlockouts                              

\usepackage{cite}
\usepackage{amsmath,amssymb,amsfonts}
\usepackage{graphicx} 
\usepackage{amsmath}
\usepackage{amssymb}  
\usepackage{mathrsfs}
\usepackage{multicol}
\usepackage{xcolor}
\usepackage{float}
\usepackage{subcaption}
\usepackage{mathtools}
\usepackage{soul}
\usepackage{bm}
\usepackage[T1]{fontenc}

\newtheorem{problem}{Problem}

\newtheorem{corollary}{Corollary}
\newtheorem{theorem}{Theorem}

\usepackage{tikz}
\usepackage{xcolor}
\usepackage{dsfont}
\usepackage{algorithm}
\usepackage{algpseudocode}
\usepackage{comment}
\usepackage{hyperref} 

\let\labelindent\relax
\usepackage[shortlabels]{enumitem}

\DeclareMathAlphabet\mathbfcal{OMS}{cmsy}{b}{n}

\definecolor{darkspringgreen}{rgb}{0.09, 0.45, 0.27}

\def\BibTeX{{\rm B\kern-.05em{\sc i\kern-.025em b}\kern-.08em
    T\kern-.1667em\lower.7ex\hbox{E}\kern-.125emX}}

\def\1{\mathds{1}}

\title{\LARGE \bf
Stability-Preserving Online Adaptation of Neural Closed-loop Maps
}

\author{Danilo Saccani, Luca Furieri, Giancarlo Ferrari-Trecate 
\thanks{This work was supported as a part of NCCR Automation, a National Centre of Competence in Research, funded by the Swiss National Science Foundation (grant number 51NF40\_225155), the Swiss National Science Foundation Ambizione (grant  number PZ00P2\_208951) and the NECON project (grant number 200021-219431).}
\thanks{D. Saccani and G. Ferrari-Trecate are with the Institute of Mechanical Engineering, Ecole Polytechnique Fédérale de Lausanne (EPFL), CH-1015 Lausanne, Switzerland. (email: \texttt{\{danilo.saccani, giancarlo.ferraritrecate\}@epfl.ch).} L. Furieri is with the Department of Engineering Science, University of Oxford, United Kingdom (\texttt{luca.furieri@eng.ox.ac.uk});}%
}

\begin{document}

\maketitle
\thispagestyle{empty}
\pagestyle{empty}
\begin{abstract}
The growing complexity of modern control tasks calls for controllers that can react online as objectives and disturbances change, while preserving closed-loop stability. Recent approaches for improving the performance of nonlinear systems while preserving closed-loop stability rely on time-invariant recurrent neural-network controllers, but offer no principled way to update the controller during operation. Most importantly, switching from one stabilizing policy to another can itself destabilize the closed-loop. We address this problem by introducing a stability-preserving update mechanism for nonlinear, neural-network-based controllers. Each controller is modeled as a causal operator with bounded $\ell_p$-gain, and we derive gain-based conditions under which the controller may be updated online. These conditions yield two practical update schemes, time-scheduled and state-triggered, that guarantee the closed-loop remains $\ell_p$-stable after any number of updates. Our analysis further shows that stability is decoupled from controller optimality, allowing approximate or early-stopped controller synthesis. We demonstrate the approach on nonlinear systems with time-varying objectives and disturbances, and show consistent performance improvements over static and naive online baselines while guaranteeing stability.
\end{abstract}

\section{Introduction}

Modern control systems operate in increasingly dynamic settings, enabled by advances in sensing, onboard computation, and model-based design. While closed-loop stability is a core requirement, alone it is not enough, as applications also demand good tracking, efficiency, and constraint handling under time-varying operating conditions. Nonlinear Optimal Control (NOC) provides a natural framework to improve performance by minimizing a task-dependent cost~\cite{bertsekas1995dynamic}. 
However, enforcing stability while optimizing a general nonlinear cost is difficult, especially for complex and time-varying dynamics~\cite{rawlings2017model,saccani2023model}. Classical approaches such as dynamic programming and the maximum principle~\cite{bertsekas1995dynamic} face significant computational limitations, and existing stability guarantees typically require restrictive assumptions on the cost, which limits applicability in practice.

Nonlinear Model Predictive Control (NMPC) offers an approximate solution to the NOC problem by implementing an implicit policy obtained from a finite-horizon problem solved at each time step in receding-horizon fashion~\cite{rawlings2017model}. While NMPC provides a structured way to handle constraints and nonlinearities, it typically imposes strict conditions on the cost to guarantee stability, which makes it difficult to apply in scenarios with rapidly changing objectives~\cite{grune2017nonlinear}. 

Neural networks (NNs) can parametrize rich families of high-performance feedback policies~\cite{cybenko1989approximation}, but certifying closed-loop stability remains a nontrivial task. Indeed, generic NN architectures do not lend themselves to the utilization of Lyapunov or small-gain conditions~\cite{sontag1993neural}.

Recent works utilize System Level Synthesis (SLS)~\cite{ho2020system,furierineuralSLS}, Internal Model Control (IMC)~\cite{saccani2024optimal,furieriPerformance}, and Youla parametrizations~\cite{wang2022youla,barbara2023learning,galimberti2024parametrizations} to improve nonlinear closed-loop performance while guaranteeing $\ell_p$-stability. The key idea is to search over a space of stable $\mathcal{L}_p$ operators used to parametrize the controller. In practice, these operators are implemented as time-invariant neural parametrizations whose stability is certified for any choice of weights. The main limitation is that the controller is fixed after training: the weights are not updated online, so significant offline training is required to cover changing operating conditions.

A key obstacle to online adaptation is that replacing one stabilizing policy with another may destabilize the transient closed loop, even when both policies are stabilizing in isolation, as known from switched-systems theory~\cite{liberzon2003switching,morse1997control,de2003switched}. 


In this work, building on the neural IMC architecture of~\cite{furieriPerformance}, we propose gain-budgeted triggering mechanisms, in both time-scheduled and state-triggered form, that certify when a newly optimized controller may safely replace the current one.

The main contributions of this work are: (i) a gain-budgeted update condition, derived from finite-gain and small-gain arguments, that guarantees the closed-loop remains $\ell_p$-stable under repeated online controller updates; (ii) two implementable triggering mechanisms, one time-based and one state-based, that enforce this condition while trading off computational effort against adaptivity; and (iii) evidence on two nonlinear benchmarks of consistent improvements over offline baselines and receding-horizon open-loop controllers.

\section*{Notation}
We denote by $\mathbb{N}$ the set of non-negative integers and $\mathbb{N}_{\geq a}$ the integers $\{a,a+1,\dots \}$.
The set of all sequences $\mathbf{v}=(v_0,v_1,v_2,\dots)$, where $v_t\in\mathbb{R}^n$ for all $t\in \mathbb{N}$, is denoted as $\ell^n$. Moreover, $\mathbf{v}$ belongs to $\ell^n_p\subset \ell^n$ with $p\in[1,\infty]$ if $\| \mathbf{v} \|_p=\left( \sum_{t=0}^{\infty} |v_t|^p \right)^{\frac{1}{p}}<\infty$, where $|\cdot |$ denotes an arbitrarily chosen vector norm. We say that $ \mathbf{v}\in\ell^n_\infty$ if $\sup_t\|v_t\|<\infty$. We refer to $v_{0:T}$ to denote the truncation of $\mathbf{v}$ with $t$ ranging from $0$ to $T$. An operator $\mathbf{A}:\ell^n \rightarrow \ell^m$ is said to be \textit{causal} if $\mathbf{A}(\mathbf{x})=(A_0(x_0), A_1(x_{0:1}),\dots,A_t(x_{0:t}),\dots)$. If $A_t(x_{0:t})=A_t(0,x_{0:t-1}),$ then $\mathbf{A}$ is said to be \textit{strictly} causal. 
An operator $\mathbf{A}:\ell^n\rightarrow \ell^m$ is $\ell_p$-stable if it is causal and $\mathbf{A}(\mathbf{a})\in \ell^m_p$ for all $\mathbf{a}\in \ell^n_p$. Equivalently, we write $\mathbf{A}\in\mathcal{L}_p$. We say that an $\mathcal{L}_p$ operator $\mathbf{A}:\mathbf{w}\mapsto\mathbf{u}$ has a finite $\mathcal{L}_p$-gain $\gamma_p(\mathbf{A})>0$ if $\|\mathbf{u}\|_p\leq\gamma_p(\mathbf{A})\|\mathbf{w}\|_p$, for all $\mathbf{w}\in\ell_p^n$. 
When $p$ is clear from the context, we simply write the gain as 
$\gamma(\mathbf{A}) > 0$ and the norm as $\|\bm{v}\|$.

\section{Preliminaries}
Let us consider the following discrete-time, time-varying nonlinear system:
\begin{equation}\label{eq:nonlinsys}
    x_t = f_{t-1} (x_{t-1},u_{t-1})+w_t, \quad t=1,2,\dots,
\end{equation}
where \( x_t \in \mathbb{R}^n \) is the state vector, \( u_t \in \mathbb{R}^m \) is the control input, and \( w_t \in \mathbb{R}^n \) represents an unknown process noise embedding the initial value of the system state by setting \( w_0 = x_0 \).
Rewriting system~\eqref{eq:nonlinsys} in operator form, we obtain:
\begin{equation} \label{eq:operatorForm}
    \mathbf{x} =\mathbf{F}(\mathbf{x},\mathbf{u})+\mathbf{w},
\end{equation}
where \( \mathbf{F}: \ell^n \times \ell^m \rightarrow \ell^n \) is a strictly causal operator defined as
$
\mathbf{F}(\mathbf{x},\mathbf{u}) = \left( 0, f_0(x_0,u_0),\dots,f_{t-1}(x_{t-1},u_{t-1}), \dots \right)
$.
We model the system as the input-to-state operator
\begin{equation}\label{eq:operatorFormIO}
\mathbfcal{F}:(\mathbf{u},\mathbf{w}) \mapsto \mathbf{x}\in \mathcal{L}_p ,
\end{equation}
and assume $\mathbfcal{F}$ is causal and admits a finite $\mathcal{L}_p$-gain $\gamma(\mathbfcal{F})<\infty$. This is satisfied when the plant is stable, or when $f_t$ represents the dynamics of a system in closed-loop with a stabilizing controller $K'$\footnote{The latter case is the one considered in the experiments in  Section~\ref{sec: numerical}.}.
Additionally, we assume that \( \mathbf{w} \in \ell_p \) and that the process noise \( w_t \) follows a distribution \( \mathcal{D} \). 
To regulate system~\eqref{eq:nonlinsys}, we consider a nonlinear, time-varying state-feedback controller of the form:
\begin{equation}\label{eq:policy}
    \mathbf{u}=\mathbf{K}(\mathbf{x})=(K_0(x_0),K_1(x_{0:1}),\dots,K_t(x_{0:t}),\dots)\,.
\end{equation}
where \( \mathbf{K}: \ell^n \rightarrow \ell^m \) is a causal operator to be designed.
Due to causality, each disturbance sequence \( \mathbf{w} \in \ell^n \) induces unique trajectories in the closed-loop system~\eqref{eq:nonlinsys}–\eqref{eq:policy} (see e.g.~\cite{furieriPerformance}). 
For a given system \( \mathbf{F} \) and controller \( \mathbf{K} \), we denote by \( \mathbf{\Phi}^x[\mathbf{F},\mathbf{K}] \) and \( \mathbf{\Phi}^u[\mathbf{F},\mathbf{K}] \) the induced closed-loop operators mapping \( \mathbf{w} \mapsto \mathbf{x} \) and \( \mathbf{w} \mapsto \mathbf{u} \), respectively. Thus,
\begin{equation*}
\mathbf{x} = \mathbf{\Phi^x[F,K](w)}, \quad
\mathbf{u} = \mathbf{\Phi^u[F,K](w)}, \quad \forall \mathbf{w} \in \ell^n.
\end{equation*}
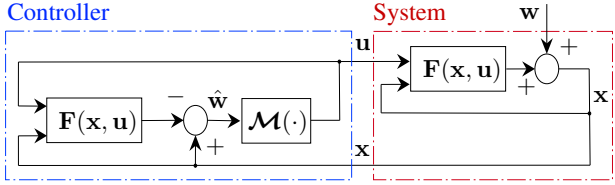
\begin{figure}
    \centering    
    \resizebox{0.95\columnwidth}{!}{%

\begin{tikzpicture}[x=0.75pt,y=0.75pt,yscale=-0.8,xscale=0.64]
\draw   (89,82) -- (168.8,82) -- (168.8,112) -- (89,112) -- cycle ;
\draw    (168.67,97.17) -- (201,97.01) ;
\draw [shift={(204,97)}, rotate = 179.73] [fill={rgb, 255:red, 0; green, 0; blue, 0 }  ][line width=0.08]  [draw opacity=0] (10.72,-5.15) -- (0,0) -- (10.72,5.15) -- (7.12,0) -- cycle    ;
\draw   (204,97) .. controls (204,91.48) and (208.48,87) .. (214,87) .. controls (219.52,87) and (224,91.48) .. (224,97) .. controls (224,102.52) and (219.52,107) .. (214,107) .. controls (208.48,107) and (204,102.52) .. (204,97) -- cycle ;
\draw    (213.9,128.57) -- (213.99,110) ;
\draw [shift={(214,107)}, rotate = 90.27] [fill={rgb, 255:red, 0; green, 0; blue, 0 }  ][line width=0.08]  [draw opacity=0] (10.72,-5.15) -- (0,0) -- (10.72,5.15) -- (7.12,0) -- cycle    ;
\draw    (65.13,87.35) -- (86,87.48) ;
\draw [shift={(89,87.5)}, rotate = 180.36] [fill={rgb, 255:red, 0; green, 0; blue, 0 }  ][line width=0.08]  [draw opacity=0] (10.72,-5.15) -- (0,0) -- (10.72,5.15) -- (7.12,0) -- cycle    ;
\draw    (65.17,107.47) -- (86.2,107.41) ;
\draw [shift={(89.2,107.4)}, rotate = 179.84] [fill={rgb, 255:red, 0; green, 0; blue, 0 }  ][line width=0.08]  [draw opacity=0] (10.72,-5.15) -- (0,0) -- (10.72,5.15) -- (7.12,0) -- cycle    ;
\draw    (65.2,127.4) -- (545.6,127.2) ;
\draw    (65.2,57.8) -- (65.13,87.35) ;
\draw    (65.2,57.8) -- (391,57.4) ;
\draw [shift={(394,57.4)}, rotate = 179.93] [fill={rgb, 255:red, 0; green, 0; blue, 0 }  ][line width=0.08]  [draw opacity=0] (10.72,-5.15) -- (0,0) -- (10.72,5.15) -- (7.12,0) -- cycle    ;
\draw    (65.17,107.47) -- (65.2,127.4) ;
\draw  [color={rgb, 255:red, 0; green, 59; blue, 255 }  ,draw opacity=1 ][dash pattern={on 7.5pt off 1.5pt on 0.75pt off 1.5pt}] (54.4,37) -- (345.2,37) -- (345.2,137) -- (54.4,137) -- cycle ;
\draw  [color={rgb, 255:red, 208; green, 2; blue, 27 }  ,draw opacity=1 ][dash pattern={on 7.5pt off 1.5pt on 0.75pt off 1.5pt}] (363.6,37.2) -- (564.4,37.2) -- (564.4,137) -- (363.6,137) -- cycle ;
\draw   (394.8,47.4) -- (474.6,47.4) -- (474.6,77.4) -- (394.8,77.4) -- cycle ;
\draw    (369.97,72.87) -- (391,72.81) ;
\draw [shift={(394,72.8)}, rotate = 179.84] [fill={rgb, 255:red, 0; green, 0; blue, 0 }  ][line width=0.08]  [draw opacity=0] (10.72,-5.15) -- (0,0) -- (10.72,5.15) -- (7.12,0) -- cycle    ;
\draw    (369.97,72.87) -- (370,92) ;
\draw   (253,82.37) -- (310.8,82.37) -- (310.8,112.37) -- (253,112.37) -- cycle ;
\draw    (475.17,62.47) -- (496.2,62.41) ;
\draw [shift={(499.2,62.4)}, rotate = 179.84] [fill={rgb, 255:red, 0; green, 0; blue, 0 }  ][line width=0.08]  [draw opacity=0] (10.72,-5.15) -- (0,0) -- (10.72,5.15) -- (7.12,0) -- cycle    ;
\draw   (499.2,62.4) .. controls (499.2,56.88) and (503.68,52.4) .. (509.2,52.4) .. controls (514.72,52.4) and (519.2,56.88) .. (519.2,62.4) .. controls (519.2,67.92) and (514.72,72.4) .. (509.2,72.4) .. controls (503.68,72.4) and (499.2,67.92) .. (499.2,62.4) -- cycle ;
\draw    (509.2,18.8) -- (509.2,49.4) ;
\draw [shift={(509.2,52.4)}, rotate = 270] [fill={rgb, 255:red, 0; green, 0; blue, 0 }  ][line width=0.08]  [draw opacity=0] (10.72,-5.15) -- (0,0) -- (10.72,5.15) -- (7.12,0) -- cycle    ;
\draw    (224,97) -- (248.6,97.36) ;
\draw [shift={(251.6,97.4)}, rotate = 180.83] [fill={rgb, 255:red, 0; green, 0; blue, 0 }  ][line width=0.08]  [draw opacity=0] (10.72,-5.15) -- (0,0) -- (10.72,5.15) -- (7.12,0) -- cycle    ;
\draw    (310.8,97) -- (335.2,97) ;
\draw    (335.2,97) -- (334.8,57.4) ;
\draw    (519.2,62.4) -- (544.8,62.4) ;
\draw    (544.8,62.4) -- (545.6,127.2) ;
\draw    (370,92) -- (544.8,92.4) ;
\draw  [fill={rgb, 255:red, 0; green, 0; blue, 0 }  ,fill opacity=1 ] (333.8,57.54) .. controls (333.8,56.96) and (334.29,56.5) .. (334.9,56.5) .. controls (335.51,56.5) and (336,56.96) .. (336,57.54) .. controls (336,58.11) and (335.51,58.58) .. (334.9,58.58) .. controls (334.29,58.58) and (333.8,58.11) .. (333.8,57.54) -- cycle ;
\draw  [fill={rgb, 255:red, 0; green, 0; blue, 0 }  ,fill opacity=1 ] (212.8,127.54) .. controls (212.8,126.96) and (213.29,126.5) .. (213.9,126.5) .. controls (214.51,126.5) and (215,126.96) .. (215,127.54) .. controls (215,128.11) and (214.51,128.57) .. (213.9,128.57) .. controls (213.29,128.57) and (212.8,128.11) .. (212.8,127.54) -- cycle ;
\draw  [fill={rgb, 255:red, 0; green, 0; blue, 0 }  ,fill opacity=1 ] (543.95,92.4) .. controls (543.95,91.83) and (544.44,91.36) .. (545.05,91.36) .. controls (545.66,91.36) and (546.15,91.83) .. (546.15,92.4) .. controls (546.15,92.97) and (545.66,93.44) .. (545.05,93.44) .. controls (544.44,93.44) and (543.95,92.97) .. (543.95,92.4) -- cycle ;
\draw (97,89) node [anchor=north west][inner sep=0.75pt]  [font=\normalsize]  {$\mathbf{F( x,u)}$};
\draw (346,112) node [anchor=north west][inner sep=0.75pt]    {$\mathbf{x}$};
\draw (187.5,74.07) node [anchor=north west][inner sep=0.75pt]    {$-$};
\draw (220.57,107.6) node [anchor=north west][inner sep=0.75pt]    {$+$};
\draw (53.93,16.93) node [anchor=north west][inner sep=0.75pt]   [align=left] {\textcolor[rgb]{0,0.16,1}{{\fontfamily{ptm}\selectfont Controller}}};
\draw (346,41) node [anchor=north west][inner sep=0.75pt]    {$\mathbf{u}$};
\draw (403,54) node [anchor=north west][inner sep=0.75pt]  [font=\normalsize]  {$\mathbf{F( x ,u)}$};
\draw (256,89) node [anchor=north west][inner sep=0.75pt]  [font=\normalsize]  {$\mathbfcal{M}( \cdot )$};
\draw (484,20) node [anchor=north west][inner sep=0.75pt]    {$\mathbf{w}$};
\draw (546.5,76) node [anchor=north west][inner sep=0.75pt]    {$\mathbf{x}$};
\draw (482,67) node [anchor=north west][inner sep=0.75pt]    {$+$};
\draw (517.98,40.63) node [anchor=north west][inner sep=0.75pt]    {$+$};
\draw (222,75) node [anchor=north west][inner sep=0.75pt]    {$\hat{\mathbf{w}}$};
\draw (361.4,16.93) node [anchor=north west][inner sep=0.75pt]   [align=left] {{\fontfamily{ptm}\selectfont \textcolor[rgb]{0.78,0,0}{System}}};
\end{tikzpicture}
}
    \caption{IMC architecture parametrizing all stability-preserving controllers via a free operator $\mathbfcal{M}\in\mathcal{L}_p$.}
    \label{fig:SLSScheme}
\end{figure}
Our objective is to design a policy $\mathbf{K}(\mathbf{x})$ addressing the following problem:
\begin{problem} \label{prb:perfBoosting}
    Design $\mathbf{K}$ by solving the infinite-horizon Nonlinear Optimal Control (NOC) problem:
    \begin{subequations} \label{eq:NOC}
    \begin{align}
        &\min_{\mathbf{K}(\cdot)}  && \mathbb{E}_{\mathbf{w}} \left[ l(\mathbf{x},\mathbf{u}) \right] \label{seq:cost}\\
        & \text{s.t.} && x_t=f_{t-1}(x_{t-1},u_{t-1})+w_t, \ \ \ w_0=x_0, \label{seq:dynP1}\\
        & &&u_t=K_t(x_{0:t}), \ \ \forall t=0,1,\dots, \label{seq:controllerP1}\\
        & && (\mathbf{\Phi^x}[\mathbf{F, K}],\mathbf{\Phi^u}[\mathbf{F, K}])\in\mathcal{L}_p\,,\label{seq:constrStab}
    \end{align}
    \end{subequations}
where $l:\ell^n\times \ell^m\rightarrow \mathbb{R}$ is any piecewise differentiable loss satisfying $l(\mathbf{x},\mathbf{u}) \geq 0$. This loss can be defined through a (possibly time-varying) stage cost, e.g., $l(\mathbf{x},\mathbf{u}) = \sum_{t=0}^\infty c_t(x_t,u_t)$, and is meant to quantify infinite-horizon closed-loop performance.
\hfill{}$\square$
\end{problem}

The expectation $\mathbb{E}_{\mathbf{w}}[\cdot]$ captures the impact of disturbances $\mathbf{w}$ on the loss, while~\eqref{seq:constrStab} imposes a hard $\ell_p$-stability constraint on the closed-loop system.

It has been shown in~\cite{furierineuralSLS,furieriPerformance} that the NOC problem admits an equivalent formulation by exploiting an IMC control architecture based on an operator $\mathbfcal{M}:\ell^n\rightarrow \ell^m$ to be designed. Specifically, all and only stability-preserving controllers, i.e., controllers verifying~\eqref{seq:constrStab}, are parametrized by $\mathbfcal{M}\in\mathcal{L}_p$.
The control architecture proposed in~\cite{furierineuralSLS,furieriPerformance} is shown in Figure~\ref{fig:SLSScheme} and incorporates a model of the system dynamics to estimate the disturbance $\mathbf{w}$. The following theorem summarizes these results:
\begin{theorem}{(adapted from Theorem 1 in~\cite{furieriPerformance})} \label{thm:neurSLS}
    Assume that the operator $\mathbfcal{F}$ is $\ell_p$-stable, and consider the evolution of~\eqref{eq:operatorForm} where $\mathbf{u}$ is defined as
    \begin{equation} \label{eq:SLSinput}
        \mathbf{u}=\mathbfcal{M}(\mathbf{x-F(x,u)})\,,
    \end{equation}
    for a causal operator $\mathbfcal{M}:\ell^n\rightarrow \ell^m$. Let $\mathbf{K}$ be the operator for which $\mathbf{u=K(x)}$ is equivalent to~\eqref{eq:SLSinput}. The following two statements hold true.
    \begin{enumerate}[label=(\roman*)]
        \item If $\mathbfcal{M}\in\mathcal{L}_p$, then the closed-loop system is $\ell_p$-stable.
        \item If there is a causal policy $\mathbf{C}$ such that $\mathbf{\Phi^x}[\mathbf{F},\mathbf{C}]$, $\mathbf{\Phi^u}[\mathbf{F},\mathbf{C}]\in\mathcal{L}_p$, then
        \begin{equation*}
            \mathbfcal{M}=\mathbf{\Phi^u}[\mathbf{F},\mathbf{C}],
        \end{equation*}
        gives $\mathbf{K}=\mathbf{C}$.
    \end{enumerate}
\end{theorem}
Theorem~\ref{thm:neurSLS} implies that exploring the space of operators $\mathbfcal{M} \in \mathcal{L}_p$ is sufficient to identify \emph{all and only stability-preserving} policies, i.e. all policies satisfying~\eqref{seq:constrStab}.
This suggests that solving Problem~\ref{prb:perfBoosting} is equivalent to replacing $\mathbf{u}=\mathbf{K}(\mathbf{x})$ in~\eqref{seq:controllerP1} with~\eqref{eq:SLSinput} and $\min_{\mathbf{K}(\cdot)}$ in~\eqref{seq:cost} with $\min_{\mathbfcal{M}\in\mathcal{L}_p}$. However, solving Problem~\ref{prb:perfBoosting} through numerical methods requires some simplifications.
First, searching within the infinite-dimensional space $\mathcal{L}_p$ is infeasible. This suggests considering subsets $\mathbfcal{O}\subset \mathcal{L}_p$ of operators $\mathbfcal{M}(\theta)$ depending on finitely many parameters $\theta\in\mathbb{R}^d$.
Various methods for defining different sets $\mathbfcal{O}$ 
are provided in~\cite{revay2023recurrent, bonassi2024structured}.
With the structures proposed in these papers, the operators $\mathbfcal{M}$ can be modelled using time-invariant nonlinear dynamical systems embedding NNs.
The formulations in~\cite{revay2023recurrent,massai2025free,zakwan2024neural} also allow imposing an explicit upper bound on the operator's $\ell_2$-gain, a feature we will exploit in Section~\ref{sec:SwitchingTheorem} for control design. We highlight that the resulting operators can be written as $\mathbfcal{M}(\theta,\bm{w})=(\mathcal{M}_0(\theta,w_0),\dots,\mathcal{M}_t(\theta,w_{0:t}),\dots)$, where $\mathcal{M}_t$ is obtained by unrolling a time-invariant dynamical system.
In combination with this parametrization, to make the NOC problem~\eqref{eq:NOC} tractable, \cite{furieriPerformance} proposes the following finite-dimensional optimization problem:
\begin{subequations} \label{eq:unconstrNOC}
    \begin{align}
        &\min_{\theta\in\mathbb{R}^d}  && \frac{1}{S} \sum^S_{s=1} \tilde{l}(x^s_{0:T}, u^s_{0:T})\label{seq:cost_unconstrNOC} \\
        & \text{s.t.} && x_t^s=f_{t-1}(x^s_{t-1},u^s_{t-1})+w^s_t, \ \ \ w^s_0=x^s_0, \nonumber\\
        & &&u^s_t=\mathcal{M}_t(\theta,w_{0:t}^s), \ \ \forall t=0,\dots,T.
    \end{align}
    \end{subequations}
In~\eqref{eq:unconstrNOC}, compared to~\eqref{seq:cost} the expectation is replaced with an empirical average over $S$ samples and the horizon is truncated at $T$. 
Crucially, \emph{stability is decoupled from optimality}; indeed, from part (i) of Theorem~\ref{thm:neurSLS}, closed-loop $\ell_p$ stability is guaranteed even when~\eqref{eq:unconstrNOC} provides an approximated local minimizer.

Although effective, this approach suffers from key limitations. First, no matter how large $T$ is chosen, it may still be insufficient for time-varying costs. Second, the existing operator parametrizations leverage time-invariant dynamics, associated to a single value of the parameters, significantly limiting the ability of the closed-loop behavior to adapt to changing costs and environmental conditions.

To address these points, in this paper, we investigate the following question: under which conditions can we safely replace the current controller with a newly updated controller $\mathbfcal{M}^{(i)}\in \mathbfcal{O}$, $i=1,2,\dots$ without compromising stability?  Answering this question enables performing online optimization to continuously update control policies.

\section{Stability-Preserving Controller Updates via triggering conditions} \label{sec:SwitchingTheorem}
\noindent
We propose to design a new class of time-varying controllers obtained by concatenating, over time, operators $\mathbfcal{M}(\theta)\in \mathbfcal{O}\subseteq \mathcal{L}_p$, where each $\mathbfcal{M}(\theta)$ is time-invariant and parametrized by a finite-dimensional parameter vector $\theta\in\mathbb{R}^d$. Our goal is to derive conditions under which the resulting switched closed loop remains $\ell_p$-stable.

We denote with $ t_i \in \mathbb{N} $, $ i = 1,2,\ldots $ the (possibly infinitely many) discrete time instants at which the operator $\mathbfcal{M}(\theta)$ is updated. 
Given an update schedule $\{t_i\}_{i \ge 0}$ and corresponding operators $\{\mathbfcal{M}^{(i)}\}_{i \ge 0}$, we define the resulting time-varying controller as $\mathbf{u}=\tilde{\mathbfcal{M}}(\mathbf{x}, \mathbf{w})$, where for any $t \in [t_i, t_{i+1})$ the composite operator $\tilde{\mathbfcal{M}}$ is defined as
\[
(\tilde{\mathbfcal{M}}(\mathbf{x}, \mathbf{w}))_t
\;:=\;
\mathcal{M}^{(i)}_{t-t_i}(\theta^{(i)},\mathbf{z}^{(i)}_{0:t-t_i}),
\]
where $\mathbf{z}^{(i)} := (x_{t_i}, w_{t_i+1}, w_{t_i+2}, \dots)$ and the operator $\mathbfcal{M}^{(i)}$ is (re)initialized with a zero internal state at $t_i$. This generalizes the $w_0=x_0$ convention: the first input is the state at $t_i$, followed by disturbances. We call $\tilde{\mathbfcal{M}}$ the \emph{concatenation} of the $\mathbfcal{M}^{(i)}$'s. 
At each update step, a new control law
, to be used over the interval $t_i:t_{i+1}-1$, is obtained by solving the optimization problem:
\begin{align} \label{eq:unconstrNOCRH}
    &\min_{\theta^{(i)}\in\mathbb{R}^d}  && \frac{1}{S} \sum^S_{s=1} \tilde{l}(x^s_{t_i:t_i+H}, u^s_{t_i:t_i+H}) \\
    & \text{s.t.} && x^s_{t_i} = x_{t_i}, \ \ z^s_{t_i} = x_{t_i},\nonumber \\
    & && x_{t_i+k}^s=f_{t_i+k-1}(x^s_{t_i+k-1},u^s_{t_i+k-1})+w^s_{t_i+k}, \nonumber\\
    & && z_{t_i+k}^s=w^s_{t_i+k}, \ \  k = 1,\dots,H, \ s=1,\dots,S, \nonumber\\
    & && u^s_{t_i+k}=\mathcal{M}_k^{(i)}(\theta^{(i)},z_{t_i:t_i+k}^s), \ \ k = 0,\dots,H, \nonumber \\ &  && \qquad \qquad \qquad \qquad\qquad \qquad \ \ s=1,\dots, S, \nonumber
\end{align}
where the cost $\tilde{l}$ is the truncation to horizon $H \in \mathbb{N}$ of the infinite-horizon cost $l(\mathbf{x},\mathbf{u})$.
Although solving problem~\eqref{eq:unconstrNOCRH} at any given time $t_i$ yields an operator $\mathbfcal{M}^{(i)}(\theta^{(i)})\in\mathbfcal{O}\subseteq\mathcal{L}_p$, which would guarantee $\ell_p$-stability if applied from $t_i$ onward, repeatedly solving~\eqref{eq:unconstrNOCRH} and varying the controller over time results in a \textit{switched closed loop}. It is well known that switching between individually stabilizing controllers can itself destabilize the closed loop~\cite{liberzon1999basic, lin2009stability, heemels2012introduction}, which necessitates a principled update mechanism.
To make online updating safe, we introduce a structured update mechanism. The idea is to allow a new operator $\mathbfcal{M}^{(i)}$ to take effect only when a prescribed triggering condition is satisfied. 
As desired, this will ensure that the closed loop induced by the concatenated controller $\tilde{\mathbfcal{M}}$ remains $\ell_p$-stable, even though the controller is updated over time through the sub-operators $\mathbfcal{M}^{(i)}(\theta)\in\mathbfcal{O}$.
We consider triggering instants $t_i \in \mathbb{N}$ defined by
\begin{equation} \label{eq:triggering_time}
    t_0=0, \quad t_{i+1}=t_{i}+\mu(x_{t_i}), \quad i\in\mathbb{N}_{\geq 0},
\end{equation}
where $\mu: \mathbb{R}^n \to \mathbb{N}_{\geq 1}$ is a triggering function to be designed for ensuring closed-loop stability. Our main result is as follows.
\begin{theorem} (Stable update of the controller)\label{thm:update}
    For a given $\bar{\gamma}$, let $ \mathbfcal{M}^{(i)} \in \mathbfcal{O} $ be the control operators updated at the triggering instants $ t_i $ given by~\eqref{eq:triggering_time}, and verifying $ \gamma(\mathbfcal{M}^{(i)}) \leq \bar{\gamma}, \  \forall i \geq 0$. If the $\mathcal{L}_p$ gains $\gamma(\mathbfcal{M}^{(i)})$ satisfy
    \begin{equation}\label{eq:bound_M_eps}
        \gamma(\mathbfcal{F}) \cdot (\gamma(\mathbfcal{M}^{(i)}) + 1) \epsilon^{(i)} \leq r^{(i)}, \quad \forall i \geq 1,
    \end{equation}
    for some non-negative scalar sequence $\bm{r}= \{r^{(i)}\}_{i=1}^\infty \in \ell_p $ and thresholds $ \epsilon^{(i)}>0 $, and if the controller update is triggered by 
    \begin{equation}
    \label{eq:switching_condition}
        |x_{t_i} | \leq \epsilon^{(i)},
    \end{equation}
    then, for every disturbance sequence $\mathbf w\in\ell_p$, the closed-loop trajectories generated by the concatenated controller $\tilde{\mathbfcal M}$ satisfy $\mathbf x\in\ell_p$ and $\mathbf u\in\ell_p$. In particular, the resulting time-varying closed loop is $\ell_p$-stable.
\end{theorem}
The proof is reported in Appendix~\ref{app:thm1}.

Theorem~\ref{thm:update} highlights a trade-off between the update condition, which bounds the state at switching times, and the intensity of the control action, governed by the gain of $\mathbfcal{M}^{(i)}$.
Indeed,~\eqref{eq:bound_M_eps} yields
\[
\epsilon^{(i)} \;\le\; \frac{r^{(i)}}{\gamma_p(\mathbfcal F)\,(\gamma_p(\mathbfcal M^{(i)})+1)}\,,
\]
and, therefore, the looser the update triggers (larger $\epsilon^{(i)}$), the gentler the controllers (smaller $\gamma_p(\mathbfcal M^{(i)})$). Conversely, larger controller gains amplify both state and input on each inter-update window.
Note that~\eqref{eq:switching_condition} is an admissibility condition: an update at time $t_i$ is only allowed if the condition is met. We do not require $t_i$ to be the first time this condition becomes true; any $t_i$ satisfying~\eqref{eq:switching_condition} is permissible. In the next section, we show how to select $\{r^{(i)}\}_{i\geq1}$, $\{\epsilon^{(i)}\}_{i\geq1}$, and the triggering mechanism \eqref{eq:triggering_time} so as to manage this trade-off.

\section{Implementation}

In this section, we derive practical update rules from Theorem~\ref{thm:update}. An update is admissible if the state at switching is sufficiently small and the gain of the new operator is compatible with that state through~\eqref{eq:bound_M_eps}--\eqref{eq:switching_condition}. The key design quantities are the thresholds $\epsilon^{(i)}$ and the sequence $\mathbf r=\{r^{(i)}\}$: larger thresholds make updates easier to trigger but force smaller gains, whereas smaller thresholds allow more aggressive controllers but delay the update.

We consider a time-scheduled scheme and a state-triggered scheme. Both require an upper bound $\hat{\gamma}(\mathbfcal F)$ on the gain of the stable or pre-stabilized plant; conservatism only makes updates more cautious. In the $\ell_2$ case, such a bound can be obtained from a storage inequality under the convention $w_0=x_0$. Since stability is enforced by the gain condition rather than by optimality, problem~\eqref{eq:unconstrNOCRH} may be solved only approximately, for instance using shorter horizons, fewer samples, or budgeted warm-started gradient steps. A practical discussion on choosing $\mathbf r$ can be found in Appendix~\ref{app:implementation}.

\subsection{Algorithm 1: Predefined Update Times}

We first consider a time-scheduled update mechanism, where $\mu$ in~\eqref{eq:triggering_time} is assumed to be constant, in which the designer selects (i)~an update period $t_{\mathrm{opt}}\in\mathbb{N}_{\ge1}$ and (ii)~a nonnegative budget sequence $\{r^{(i)}\}_{i\ge1}\in\ell_p$.
Updates are only attempted every $t_{\mathrm{opt}}$ steps:
\[
t_0 = 0, \qquad
t_{i+1} = t_i + t_{\mathrm{opt}}, \quad i \in \mathbb{N}_{\geq1}.
\]
At each scheduled time $t_i$, we (a) measure the current state and set $\epsilon^{(i)} := |x_{t_i}|$, and (b) check if an update is admissible under Theorem~\ref{thm:update}.
Rearranging \eqref{eq:bound_M_eps} yields an upper bound on the allowable gain of the new controller:
\begin{equation}\label{eq:Mi_gain_bound_algo1}
0 < \gamma_p(\mathbfcal{M}^{(i)}) \;\le\;
\frac{r^{(i)}}{\gamma_p(\mathbfcal{F}) \, \epsilon^{(i)}} - 1.
\end{equation}
We also require $\gamma_p(\mathbfcal{M}^{(i)}) \le \bar{\gamma}$, where $\bar{\gamma}$ is the uniform gain bound assumed in Theorem~\ref{thm:update}.
If at $t_i$ there is $ \gamma_p(\mathbfcal{M}^{(i)})$ satisfying~\eqref{eq:Mi_gain_bound_algo1}, we solve~\eqref{eq:unconstrNOCRH} (with that gain bound enforced) to obtain new parameters $\theta^{(i)}$ for $\mathbfcal{M}^{(i)}$, and we apply this controller on $[t_i,\dots,t_{i+1}{-}1]$. If it is infeasible (e.g., because $|x_{t_i}|$ is too large or $r^{(i)}$ is too small), we keep the previous controller for the next interval, $\mathbfcal{M}^{(i)}=\mathbfcal{M}^{(i-1)}$, which preserves $\ell_p$-stability.  
Whenever there is $ \gamma_p(\mathbfcal{M}^{(i)})$ satisfying~\eqref{eq:Mi_gain_bound_algo1}, Theorem~\ref{thm:update} guarantees that we can update the controller again and the resulting time-varying closed-loop remains $\ell_p$-stable.

In summary, $t_{\mathrm{opt}}\in\mathbb{N}_{\geq 1}$ is freely chosen (with no minimum dwell-time constraint) and determines how often we attempt an update and $r^{(i)}$ sets how aggressive the $i$-th update is allowed to be: for a given state norm $|x_{t_i}|$, larger $r^{(i)}$ allows larger controller gain and thus potentially faster transients.

\subsection{Algorithm 2: State-Dependent Triggering}

The second scheme removes the fixed update period. Instead, the designer specifies (i) an admissible gain level $\gamma_p(\mathbfcal{M})$ that upper-bounds the $\mathcal{L}_p$-gain of any controller to be used, and (ii) a nonnegative sequence $\{r^{(i)}\}_{i\ge 1} \in \ell_p$. 

Given $\gamma_p(\mathbfcal{M})$ and $r^{(i)}$, we define the admissible state threshold for the $i$-th update as
\[
\epsilon^{(i)} \;=\; \frac{r^{(i)}}{\gamma_p(\mathbfcal{F}) \big(\gamma_p(\mathbfcal{M}) + 1\big)},
\]
which is exactly the feasibility condition of Theorem~\ref{thm:update}.
We then set $t_i$ to be the first time after $t_{i-1}$ such that
\[
|x_{t_i}| \;\le\; \epsilon^{(i)}.
\]
At that instant $t_i$, we solve~\eqref{eq:unconstrNOCRH} to obtain a new controller $\mathbfcal{M}^{(i)}$ with $\gamma_p(\mathbfcal{M}^{(i)}) \le \gamma_p(\mathbfcal{M})$, and keep it active until the next trigger time $t_{i+1}$.
This ``self-triggered'' rule inverts Algorithm 1. Rather than updating on a fixed schedule and then checking if the state is small enough, we wait until the state is small enough to admit an update of the operator $\mathbfcal{M}$.
This can reduce unnecessary recomputations when the state is still large.
The only design degrees of freedom are therefore the gain bound $\gamma_p(\mathbfcal{M})$ and the budget sequence $\{r^{(i)}\}$, which together determine when updates occur. Again, enforcing the gain cap $\gamma_p(\mathbfcal{M}^{(i)}) \le \gamma_p(\mathbfcal{M})$ and trigger condition implies $\ell_p$-stability by Theorem~\ref{thm:update}.

\section{Numerical Example}
\label{sec: numerical}
In this section, we demonstrate the application of the proposed approach to robotic tasks. 
First, we apply our solution to the \textit{mountains} problem from~\cite{furieriPerformance,onken2021neural} to highlight the advantages of our method.
Then, we present a tracking problem in a time-varying environment with dynamic obstacles. 

We consider a point-mass system of mass $m\in\mathbb{R}_+$ subject to nonlinear friction forces:
\begin{equation} \label{eq:example}
    \begin{bmatrix}
    q_{t+1} \\
    v_{t+1}
\end{bmatrix} =
\begin{bmatrix}
    q_t + T_s v_t \\
    v_t + T_s m^{-1} \bigl( F_t + G(v_t) \bigr)
\end{bmatrix} + w_t,
\end{equation}
where $q_t \in \mathbb{R}^2$ and $v_t \in \mathbb{R}^2$ represent the position and velocity of the system, respectively, while $F_t \in \mathbb{R}^2$ is the control input. The sampling time is denoted as $T_s = 0.05$ seconds. The term $w_t \in \mathbb{R}^4$ represents process noise. The nonlinear damping term is defined as:
$
G(v_t) = -b_1 v_t + b_2 \tanh(v_t),
$
where $0 < b_2 < b_1$. 
To parameterize the operators $\mathbfcal{M}^{(i)} \in \mathbfcal{O}$, we employed a Recurrent Equilibrium Network (REN)~\cite{revay2023recurrent} with prescribed gain. Problem~\eqref{eq:unconstrNOCRH} is solved using backpropagation through time to compute gradients (see~\cite{furieriPerformance}), and the Adam optimizer to update the parameters via stochastic gradient descent.
The approach is implemented using PyTorch, and the code to reproduce the examples is available at: \url{https://github.com/DecodEPFL/Online_neurSLS.git}.
\begin{figure*}[]
    \centering
    \begin{subfigure}{0.39\textwidth}
        \centering \includegraphics[width=\textwidth]{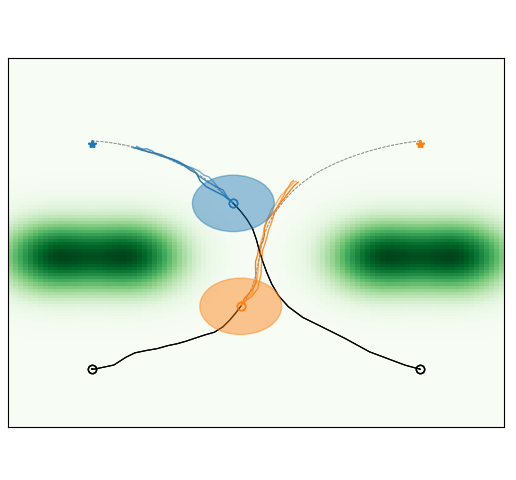}
    \end{subfigure}
    \begin{subfigure}{0.48\textwidth}
        \centering
        \includegraphics[width=\textwidth]{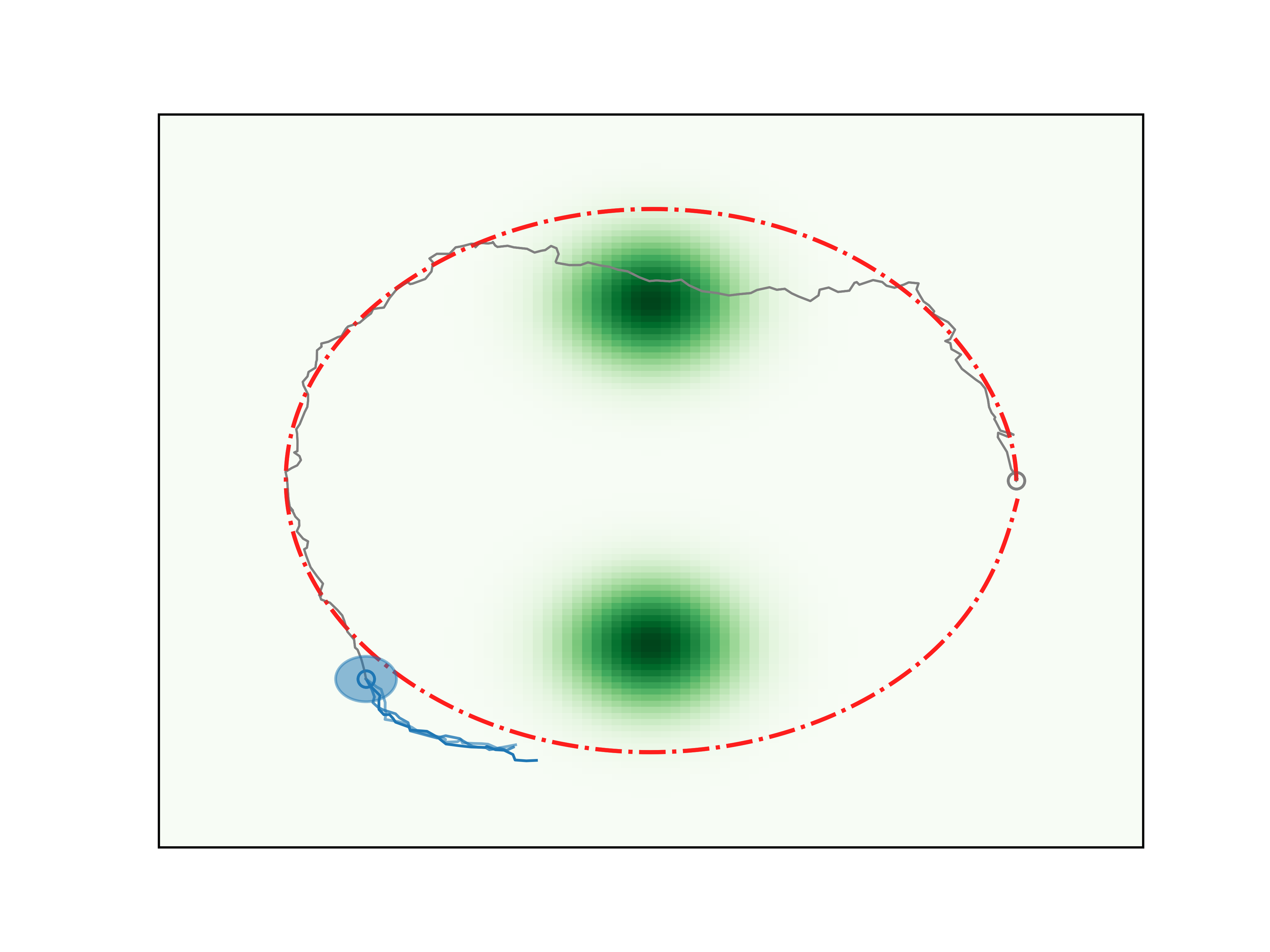}
    \end{subfigure}
    \caption{Qualitative closed-loop behavior of Algorithm~1 in the two case studies. (Left) Mountains problem at $\tau = 0.8\,\mathrm{s}$. Green: obstacles. Colored lines: predicted trajectories over $[\tau,\tau+1.25]$; black: executed trajectory over $[0,\tau)$; gray dashed: predicted continuation after $\tau$. Colored disks show agent positions and radii. (Right) Dynamic-obstacles problem at $\tau_1 = 9.3\,\mathrm{s}$. Gray: executed trajectory; red dash-dot: reference; blue: predicted trajectory over $[\tau_1,\tau_1+0.5]$; green: obstacle positions at $\tau_1$.}
    \label{fig:trackingCL}
\end{figure*}
\subsection{Mountains problem}
The \textit{mountains} scenario depicted in Figure~\ref{fig:trackingCL} features two agents aiming to collaboratively navigate through a narrow valley. Each agent $j\in\{ 1,2\}$ evolves according to~\eqref{eq:example} and has a target position $\bar{q}_j \in \mathbb{R}^2$ that must be reached with zero velocity (i.e., $\bar{v}_j = 0_2$). We stack both agents' states, so the overall state is $x_t=[q_{t,1},v_{t,1},q_{t,2},v_{t,2}]^\top\in\mathbb{R}^8$ and $\bar{x}$ stacks the target positions and zero velocities.
We study the regulation problem in error coordinates with respect to the desired equilibrium $\bar{q}_j,\bar{v}_j$. We apply a baseline linear feedback $K' = \operatorname{diag}(k_1, k_2)$ with gains $k_1, k_2 > 0$, around the equilibrium, and we set $F_{t,j}=K'(\bar q_j -q_{t,j})+u_{t,j}$, where $u_{t,j}$ is an auxiliary input. We treat $u_t=[u_{t,1},u_{t,2}]^\top$ as the control signal to be optimized. Under this pre-stabilization, the closed-loop error dynamics can be written in the form of~\eqref{eq:nonlinsys}, with state $e_t=x_t-\bar{x}\in\mathbb{R}^8$, input $u_{t}=[u_{t,1}, u_{t,2}]^\top\in\mathbb{R}^4$, and disturbance $w_t\in\mathbb{R}^8$. The associated input-to-state operator $\mathbfcal{F}_e(\bm{u},\bm{w})\mapsto\bm{e}$ has finite $\ell_2$-gain. 
An upper bound on this gain is obtained from a discrete-time bounded-real LMI~\cite{de1992discrete}; see Appendix~\ref{app:ell2gain}.
In this example, we set $p=2$, and we model the disturbance sequence $\mathbf{w}$ as Gaussian noise with exponential decay.
The loss function is defined as:
\begin{align}
    \tilde l(e_{0:T},u_{0:T}) &= \sum^{T}_{t=0} \big(l_{\text{traj}}(e_t, u_t) + l_{\text{ca}}(e_t) + l_{\text{obs}}(e_t)\big), \nonumber
\end{align}
where $l_{\text{traj}}(e, u) = [e^\top \ u^\top] Q [e^\top \ u^\top]^\top$ with $Q \succeq 0$ penalizing deviation from the targets and control effort. The terms $l_{\text{ca}}(e)$ and $l_{\text{obs}}(e)$ are barrier functions that prevent collisions between agents and obstacles (see~\cite{furieriPerformance} for their definitions). We compare with the approach proposed in~\cite{furieriPerformance}, where the controller is parametrized by an operator $\mathbfcal{M}\in\mathcal{L}_2$ obtained offline over a horizon $T=100$, matching the total simulation time used for our online controller. Figure~\ref{fig:trackingCL} (left) illustrates closed-loop trajectories obtained using Algorithm 1, which attempts an update to a new operator $\mathbfcal{M}^{(i)}$ at fixed intervals of $t_{\text{opt}} = 2$ time steps. Problem~\eqref{eq:unconstrNOCRH} is solved online to update the operator with a horizon of $H = 25$ and $S = 3$ disturbance realizations. Two scenarios are simulated: the nominal scenario replicating the example in~\cite{furieriPerformance}, and a perturbed scenario involving abrupt impulse disturbances. Specifically, the local disturbance $w_{t,j}$ is replaced by: 
\begin{equation} \label{eq:delta}
    w_{t,j} + 0.3\delta(t-1)+0.3\delta(t-8), \quad j=1,2,
\end{equation}
where $\delta$ is the Kronecker delta. 
Over 50 runs, our method reduces the average cost by 35.1\% (nominal) and 40.3\% (perturbed) versus~\cite{furieriPerformance}, indicating improved adaptation to abrupt, previously unseen disturbance events; boxplots are reported in Appendix~\ref{app:example}.

\subsection{Dynamic obstacles}
In the second scenario (Figure~\ref{fig:trackingCL}), an agent tracks a circular path while avoiding dynamic obstacles. 
We track a reference trajectory $x_{\text{ref},t}$ 
by applying a baseline stabilizing feedback $K'$ together with a feedforward term $u_{\text{ref},t}$. We define the tracking error $e_t=x_t-x_{\text{ref},t}\in\mathbb{R}^4$ and write the total input as $ F_t = u_{\text{ref},t} - K' e_t + u_{t}$, where $u_t$ is an auxiliary control signal that we optimize online. Under this pre-stabilization, the closed-loop error dynamics for $e_t$ can be written as a nonlinear system $e_{t} = f_e(e_{t-1}, u_{t-1}) + w_{t}$, which matches the general form of~\eqref{eq:nonlinsys} with state $e_t$, input $u_t$, and disturbance $w_t$.
As in the previous example, the gain of the associated input-to-state operator $\mathbfcal{F}_e(\bm{u},\bm{w})\mapsto\bm{e}$ is computed upper-bounding the worst-case realization of the nonlinearity.
Two dynamic obstacles move on the $y$-axis following a perturbed sinusoidal trajectory:
\begin{align} \label{eq:obsmotion}
& y^{\text{obs}}_{t,j} = A \sin((2\pi \psi_j + \eta_j)t + \phi_j) + y^{\text{obs}}_{0,j}, \quad j=1,2.
\end{align}
where $\psi_j$ is the nominal frequency, $\eta_j \sim \mathcal{N}(0,0.01)$ is a random frequency perturbation, and $\phi_j$ is a random phase.
The finite-horizon objective $\tilde l$ in \eqref{eq:unconstrNOCRH} is now time-varying because the obstacles move. We define
\begin{align}
    \tilde l(e_{0:T},u_{0:T}) &= \sum^{T}_{t=0} \big( l_{\text{traj}}(e_t, u_t) + l_{\text{obs}}(e_t, t)\big). \nonumber
\end{align} 
where $l_{\text{traj}}$ is a positive semidefinite quadratic penalty on tracking error and control effort, and $l_{\text{obs}}(e_t, t)$ penalizes collisions with the moving obstacle.
The system is simulated for $300$ time steps under a disturbance sequence generated at each step with bounded amplitude and interpreted as an element of $\ell_2$ by zero extension beyond the simulation horizon. In other words, the disturbance is persistent over the considered horizon, but has finite energy on the full infinite sequence. Accordingly, in this experiment we set $p=2$, consistently with the certified $\ell_2$-gain of the REN controller class. We compare against (i) the offline controller of~\cite{furieriPerformance} and (ii) 
a receding-horizon open-loop (RHO) planner that solves at time $t$ the problem
\begin{align*}
    &\min_{u_{0:H-1}} \ \sum_{s=0}^{S} \tilde l(e_{0:H},u_{0:H}) \\
    \text{s.t.}&\quad e_{k+1}^s=f_e(e_k^s,u_k)+w_k^s,\ e^s_0=e_t, \ \forall s, \ \mathbf{w}^{s} \sim {\mathcal{D}},
\end{align*}
and applies the first input $u_t=u^*_0$, advances the state, and re-solves at $t{+}1$.
For a fair comparison, all approaches use the same horizon ($H=10$) and cost weights. Algorithm~1 attempts updates $\mathbfcal{M}$ every $t_{\mathrm{opt}}=1$ time step, using a budgeted, inexact solution (100 epochs of gradient descent warm-started from the previous parameters). This is admissible under Theorem~\ref{thm:update} because stability is guaranteed even without reaching local optimality. The right panel in Figure~\ref{fig:trackingCL} shows the closed-loop trajectory (in grey) and the predicted trajectories (in blue). 
Persistent disturbances make perfect tracking unachievable, and apparent overlaps between moving obstacles and past trajectory segments are only visualization artifacts. Over 50 runs, the proposed approach achieves lower cost than both RHO and the offline controller from~\cite{furieriPerformance}, while also exhibiting smaller variability; the corresponding boxplots are reported in Appendix~\ref{app:example}. Overall, the results indicate that online adaptation with closed-loop prediction improves both performance and reliability relative to the two baselines.

\section{Conclusions}
We introduced a framework for iteratively updating neural-network-based controllers while maintaining closed-loop stability via gain-based update conditions. This opens several research directions. First, robustness margins against model-plant mismatch could be embedded directly into the update logic, preserving stability under imperfect models. Second, online optimization could be used to update both the control operator and the system model used by the controller. Finally, connecting these ideas to reinforcement learning by enforcing our gain-based admissibility condition during policy improvement may offer stability guarantees during learning, not only after convergence.

\bibliographystyle{IEEEtran}
\bibliography{biblio}    

\newpage
\appendix
\subsection{Proof of Theorem~\ref{thm:update} and Input to State stability relation}
\label{app:thm1}
We first establish an upper bound for the state and 
input norms when applying $\mathbfcal{M}^{(0)}$ on the 
first interval $[t_0, t_1-1]$, where $t_0=0$.
From the definition of $\mathbfcal{F}$ in \eqref{eq:operatorForm} 
and the controller $\tilde{\mathbfcal{M}}$, the controller 
input is $\mathbf{z}^{(0)}=(x_0, w_1, \dots)$. From~\eqref{eq:nonlinsys}, one has $w_0=x_0$, so 
$\mathbf{z}^{(0)} = \mathbf{w}$. 
Let $\bar{x}_0 = x_{0:t_1-1}$ and $\bar{w}_0 = w_{0:t_1-1}$.
\begin{align}
    \| \mathbf{u}_{0:t_1-1} \| &\leq \gamma(\mathbfcal{M}^{(0)}) \|\bar{w}_0\| \label{eq:input_bound_app} \\
    \| \bar{x}_0 \| &\leq \gamma(\mathbfcal{F}) (\| \mathbf{u}_{0:t_1-1} \| + \| \bar{w}_0 \|) \label{eq:state_bound_int_app}
\end{align}
Substituting \eqref{eq:input_bound_app} into 
\eqref{eq:state_bound_int_app} yields:
\begin{equation} \label{eq:state_bound_app}
    \| \bar{x}_0 \| \leq \underbrace{\gamma(\mathbfcal{F}) (\gamma(\mathbfcal{M}^{(0)}) + 1)}_{C_0} \|\bar{w}_0\|.
\end{equation}
Now consider the evolution on a subsequent window $i \ge 1$ 
on the interval $[t_i, t_{i+1}-1]$. 
Let $\bar{x}_i = x_{t_i:t_{i+1}-1}$, $\bar{u}_i = u_{t_i:t_{i+1}-1}$, 
$\bar{w}_i = w_{t_i:t_{i+1}-1}$, and $\bar{z}_i = (x_{t_i}, w_{t_i+1}, 
\dots, w_{t_{i+1}-1})$. 
We obtain
\begin{equation} \label{eq:proof_iss_bound}
    \|\bar{x}_i\|_p \le\footnote{The system $\mathbfcal{F}$ on this window has a non-zero 
initial state $x_{t_i}$ and external inputs $\bar{u}_i$ and 
$\bar{w}_i$, and since as established by the $w_0=x_0$ convention, the 
unforced response to an initial state $x_{t_i}$ is 
equivalent to the response to a disturbance impulse 
$\mathbf{w}' = (x_{t_i}, 0, \dots)$.} \gamma_p(\mathbfcal{F}) |x_{t_i}| + \gamma_p(\mathbfcal{F}) (\|\bar{u}_i\|_p + \|\bar{w}_i\|_p)
\end{equation}
The controller's output $\bar{u}_i$ is bounded by its gain with zero initial condition:
\begin{equation}\label{eq:proof_u_bound}
    \| \bar{u}_i \|_p \le \gamma_p(\mathbfcal{M}^{(i)})\,\|\bar{z}_i\|_p.
\end{equation}
The input $\bar{z}_i$ is bounded using the triangle inequality:
\begin{equation}\label{eq:proof_z_bound}
    \|\bar{z}_i\|_p \le |x_{t_i}| + \|w_{t_i+1:t_{i+1}-1}\|_p \le |x_{t_i}| + \|\bar{w}_i\|_p
\end{equation}
Substitute \eqref{eq:proof_u_bound} and \eqref{eq:proof_z_bound} 
into \eqref{eq:proof_iss_bound} and grouping terms by $|x_{t_i}|$ and $\|\bar{w}_i\|_p$ yields:
\begin{align*}
    \|\bar{x}_i\|_p &
    \leq \gamma_p(\mathbfcal{F})(\gamma_p(\mathbfcal{M}^{(i)})+1) |x_{t_i}| \\
    &\quad + \gamma_p(\mathbfcal{F})(\gamma_p(\mathbfcal{M}^{(i)})+1) \|\bar{w}_i\|_p
\end{align*}
Now, apply the triggering condition \eqref{eq:switching_condition}, 
$|x_{t_i}| \le \epsilon^{(i)}$, and the gain condition \eqref{eq:bound_M_eps}, 
$\gamma_p(\mathbfcal{F})(\gamma_p(\mathbfcal{M}^{(i)})+1) \epsilon^{(i)} \le r^{(i)}$:
\begin{align} \label{eq:upperboundSequence}
    \| \bar{x}_i \|_p 
    &\leq r^{(i)} + C \|\bar{w}_i\|_p,
\end{align}
where $C = \gamma_p(\mathbfcal{F}) (\bar{\gamma} + 1)$ (using 
$\bar{\gamma}$ as an upper bound).
We now analyze the total norm of $\mathbf{x} = (\bar{x}_0, \bar{x}_1, \dots)$.\\
\textbf{Case: $1 \leq p < \infty$}.
Summing the $p$-th power of~\eqref{eq:upperboundSequence} over all windows $i$:
\begin{align} \label{eq:sumOverWindow}
    \| {\mathbf{x}} \|_p^p &= \| \bar{x}_0 \|_p^p + \sum_{i=1}^{\infty} \| \bar{x}_{i} \|_p^p \\
    &\le (C_0 \|\bar{w}_0\|_p)^p + \sum_{i=1}^{\infty} \left( r^{(i)} + C \| \bar{w}_i \|_p \right)^p \notag
\end{align}
Let $\tilde{\mathbf{w}}=(\|\bar{w}_1\|_p,\| \bar{w}_2\|_p,\ldots)$. The sum in~\eqref{eq:sumOverWindow} is the 
$p$-norm of a sum of sequences in $\ell_p$:
\begin{align*}
    \| {\mathbf{x}} \|_p^p &\leq (C_0 \|\bar{w}_0\|_p)^p + \left( \|\mathbf{r} + C\tilde{\mathbf{w}}\|_p \right)^p \\
    &\leq (C_0 \|\bar{w}_0\|_p)^p + \left( \|\mathbf{r}\|_p + C\|\tilde{\mathbf{w}}\|_p \right)^p
\end{align*}
Using $(a+b)^p \leq 2^{p-1}(a^p+b^p)$, which holds for $p\in\mathbb{N}_{\geq1}$ on the second term:
    $\| \mathbf{x} \|_p^p  \leq (C_0 \|\bar{w}_0\|_p)^p + 2^{p-1} (\|\mathbf{r}\|_p^p + C^p\|\tilde{\mathbf{w}}\|_p^p).$
The windows $\bar{w}_i$ are disjoint. Let $\mathbf{w}' = 
(\|\bar{w}_0\|, \tilde{\mathbf{w}})$. 
Then $\|\mathbf{w}'\|_p^p = \sum_{i=0}^\infty \|\bar{w}_i\|_p^p = \|\mathbf{w}\|_p^p$. 
The total norm is bounded by:
\begin{align*}
    \| \mathbf{x} \|_p^p \le 2^{p-1} \left( (C_0 \|\bar{w}_0\|_p)^p + \|\mathbf{r}\|_p^p + C^p\|\tilde{\mathbf{w}}\|_p^p \right) \\
    \le 2^{p-1} \left( \max(C_0, C)^p \|\mathbf{w}\|_p^p + \|\mathbf{r}\|_p^p \right)
\end{align*}
Since $ \mathbf{r} , {\mathbf{w}} \in \ell_p $, the right-hand side is 
finite, so we conclude that $\| {\mathbf{x}} \|_p < \infty$. 
The same arguments applied to \eqref{eq:proof_u_bound} 
show $\|\bm{u}\|_p<\infty$. Thus, the system is $\ell_p$-stable.

\textbf{Case: $ p = \infty $}.
The bound \eqref{eq:upperboundSequence} (which also holds for $p=\infty$) becomes: $\|\bar{x}_i\|_\infty \leq r^{(i)} + C \| \bar{w}_i \|_\infty.$\\
Taking the supremum over all $i\in\mathbb{N}_{\geq 1}$:
\begin{align*}
    \sup_{i\in\mathbb{N}_{\geq 1}} \|\bar{x}_i\|_\infty &\le \sup_{i\in\mathbb{N}_{\geq 1}} r^{(i)} + \sup_{i\in\mathbb{N}_{\geq 1}} (C \|\bar{w}_i\|_\infty) 
\end{align*}
The total state norm is $\|\mathbf{{x}} \|_\infty = \max( \|\bar{x}_0\|_\infty, 
\sup_{i\in\mathbb{N}_{\geq 1}} \|\bar{x}_i\|_\infty )$.
From \eqref{eq:state_bound_app}, $\|\bar{x}_0\|_\infty \le C_0 \|\bar{w}_0\|_\infty$, and hence
\begin{equation} \label{eq:boundinf}
    \|\mathbf{{x}} \|_\infty \le \max(C_0 \|\bar{w}_0\|_\infty, \|\mathbf{r}\|_\infty + C \sup_{i\in\mathbb{N}_{\geq 1}} \|\bar{w}_i\|_\infty)
\end{equation}
Since $\sup_{i\in\mathbb{N}_{\geq 0}} \|\bar{w}_i\|_\infty \le \|\mathbf{w}\|_\infty$, and $\mathbf{r}, 
\mathbf{w} \in \ell_\infty$, the right hand side of~\eqref{eq:boundinf} is finite. Thus, 
the system is $\ell_\infty$-stable.\\
Therefore, for any $p\in[1,\infty]$ and any disturbance sequence $\mathbf w\in\ell_p$, the corresponding closed-loop trajectories satisfy $\mathbf x\in\ell_p$ and $\mathbf u\in\ell_p$. Hence, the closed loop induced by the concatenated controller $\tilde{\mathbfcal M}$ is $\ell_p$-stable.
\hfill$\blacksquare$

Next, we discuss how imposing additional conditions on $\bm r$ can yield an input-to-state type bound for the resulting closed loop, on top of $\ell_p$-stability.

\begin{corollary}[ISS under time-scheduled persistent updates]\label{cor:ISS}
Consider Algorithm~1 with $p=\infty$ and fixed update times $t_i=i\,t_{\mathrm{opt}}$, $i\in\mathbb N$, and assume that the conditions of Theorem~\ref{thm:update} hold at every scheduled update.
Suppose that
\[
r^{(i)} = d^{(i)}\,\gamma(\mathbfcal{F})\bigl(\gamma(\mathbfcal{M}^{(0)})+1\bigr)|x_0|
+ g \|\mathbf{w}\|_\infty,
\qquad i\ge 1,
\]
for some positive, strictly decreasing scalar sequence $\{d^{(i)}\}_{i\ge 0}$ such that $d^{(0)}\ge 1$ and $\lim_{i\to\infty} d^{(i)}=0$, and for some constant $g>0$. Then the resulting closed loop is Input-to-State Stable (ISS)~\cite{jiang2001input}; that is, there exist functions $\beta\in\mathcal{KL}$ and $\alpha\in\mathcal{K}_\infty$ such that, for any initial condition $x_0$ and any bounded disturbance sequence $\mathbf w\in\ell_\infty$, the state trajectory satisfies
\[
|x_t| \le \beta(|x_0|,t)+\alpha(\|\mathbf w\|_\infty), \qquad \forall t\in\mathbb N .
\]
\end{corollary}
\begin{proof}
Let
\[
D:=\gamma(\mathbfcal{F})\bigl(\gamma(\mathbfcal{M}^{(0)})+1\bigr),
\qquad
G:=\gamma(\mathbfcal{F})(\bar\gamma+1).
\]
From the proof of Theorem~\ref{thm:update}, if $t\in[t_i,t_{i+1})$ with $i\ge 1$, then
\[
|x_t|
\le
r^{(i)} + G\|\mathbf w\|_\infty
=
D\,d^{(i)}|x_0| + (g+G)\|\mathbf w\|_\infty .
\]
For $t\in[t_0,t_1)$, \eqref{eq:state_bound_app} gives
\[
|x_t|
\le
D |x_0| + G\|\mathbf w\|_\infty
\le
D\,d^{(0)}|x_0| + (g+G)\|\mathbf w\|_\infty,
\]
since $d^{(0)}\ge 1$ and $g>0$.

Now define a continuous, nonincreasing function $\bar d:[0,\infty)\to(0,\infty)$ as follows: set $\bar d(t)=d^{(0)}$ for $t\in[0,t_1]$, and for each $i\ge 1$ let $\bar d$ be linear on $[t_i,t_{i+1}]$ with
\[
\bar d(t_i)=d^{(i-1)}, \qquad \bar d(t_{i+1})=d^{(i)}.
\]
Since $\{d^{(i)}\}$ is strictly decreasing and $t_i\to\infty$, the function $\bar d$ is continuous, nonincreasing, and satisfies $\bar d(t)\to 0$ as $t\to\infty$. Moreover, for every $t\in[t_i,t_{i+1})$ one has $\bar d(t)\ge d^{(i)}$.

Therefore, for all $t\ge 0$,
\[
|x_t|
\le
D\,\bar d(t)\,|x_0| + (g+G)\|\mathbf w\|_\infty.
\]
Define
\[
\beta(s,t):=D\,\bar d(t)\,s,
\qquad
\alpha(s):=(g+G)s.
\]
Then $\alpha\in\mathcal K_\infty$. Also, for each fixed $t\ge0$, $\beta(\cdot,t)\in\mathcal K$, and for each fixed $s\ge0$, $\beta(s,\cdot)$ is continuous, nonincreasing, and converges to $0$ as $t\to\infty$. Hence $\beta\in\mathcal{KL}$, and the claimed ISS estimate follows.
\end{proof}

\subsection{Implementation details}
\label{app:implementation}
\subsubsection{Design of the sequence $\mathbf{r}$}

The sequence $\mathbf{r} = \{r^{(i)}\}_{i\ge 1}$ regulates how aggressive updates are allowed to be over time.
The design of $\mathbf{r}$ can leverage \textit{a priori} information obtained through offline closed-loop simulations. Let $\tilde{x}_{0:T}$ be a nominal state trajectory generated under a representative disturbance and a baseline controller $\mathbfcal{M}^{(0)}$. We define a time-indexed budget profile
\[
0 \le r_t = \rho \,\|\tilde{x}_t\|, \quad t=0,\dots,T,
\]
with $\rho>0$ a scaling factor. To obtain an infinite-horizon sequence compatible with Theorem~\ref{thm:update}, the profile is extended for $t>T$ by a nonnegative tail $\{r_t\}_{t>T}$ chosen so that $\{r_t\}_{t\ge0}\in \ell_p$. For instance, when $1\le p<\infty$, one may append any summable tail, e.g. $
r_t = r_T \eta^{\,t-T}$, $t>T,$ for some $\eta\in(0,1)$, while for $p=\infty$ any bounded continuation is admissible (and a decaying continuation can be used if one also wants the update budget to vanish asymptotically). Online, at each actual update time $t_i$, we set $r^{(i)} := r_{t_i}$.

In this way, large nominal deviations $\|\tilde{x}_t\|$ early in the system evolution yield larger $r^{(i)}$, allowing updates even for relatively large states (or higher controller gains), while smaller $\|\tilde{x}_t\|$ later in the system evolution yield smaller $r^{(i)}$, enforcing more conservative updates near the origin. This provides a simple, simulation-driven procedure to generate a sequence $\mathbf{r}$ consistent with Theorem~\ref{thm:update}.

\subsubsection{Controller Updates through optimization}

The method for updating the controller depends significantly on the chosen update strategy and the available computational resources. Indeed, solving problem~\eqref{eq:unconstrNOCRH} can be computationally demanding. 
A practical way to reduce computational burden is to solve a simplified instance of \eqref{eq:unconstrNOCRH} (shorter horizon
$H$, fewer disturbance samples). Equivalently, one may warm-start from the previously found parameters $\theta$ and take a budgeted number of gradient steps during the sampling interval; the resulting inexact iterate can be deployed without risking destabilization.
Indeed, \emph{the proposed method ensures stability without requiring optimality} and without requiring specific assumptions on the cost function, except that it must be differentiable (see Problem~\ref{prb:perfBoosting}).

\subsection{Numerical Example}
\label{app:example}
Figures~\ref{fig:boxplotMountain} and~\ref{fig:boxplotTimevarying} report the boxplots of the total closed-loop cost for the mountains and dynamic-obstacles scenarios, respectively.
\begin{figure}[t]  
        \centering        \includegraphics[width=\columnwidth]{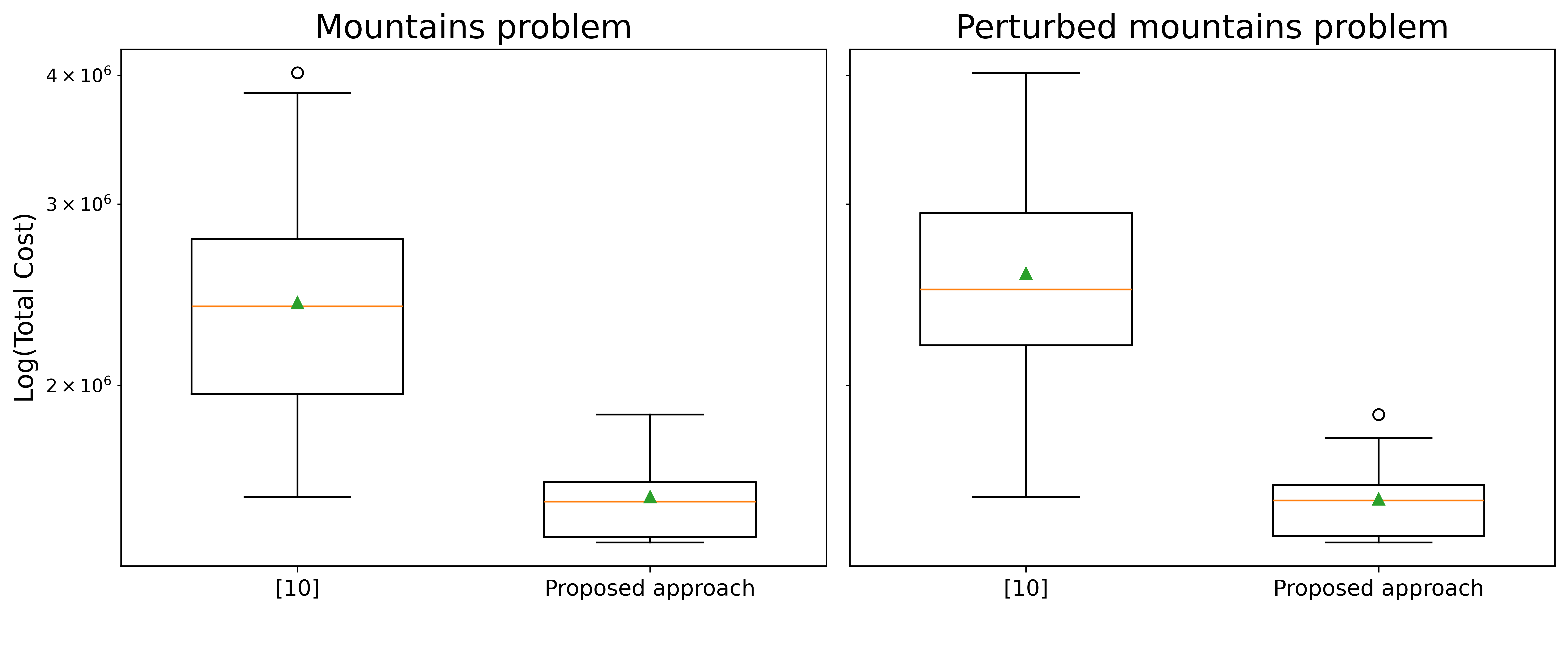}
 \caption{Mountains problem - Total cost (log scale) over 50 runs: nominal scenario (center) and perturbed scenario with impulse disturbances~\eqref{eq:delta} (right).}
    \label{fig:boxplotMountain}
\end{figure}
\begin{figure}[t]  
        \centering        \includegraphics[width=\columnwidth]{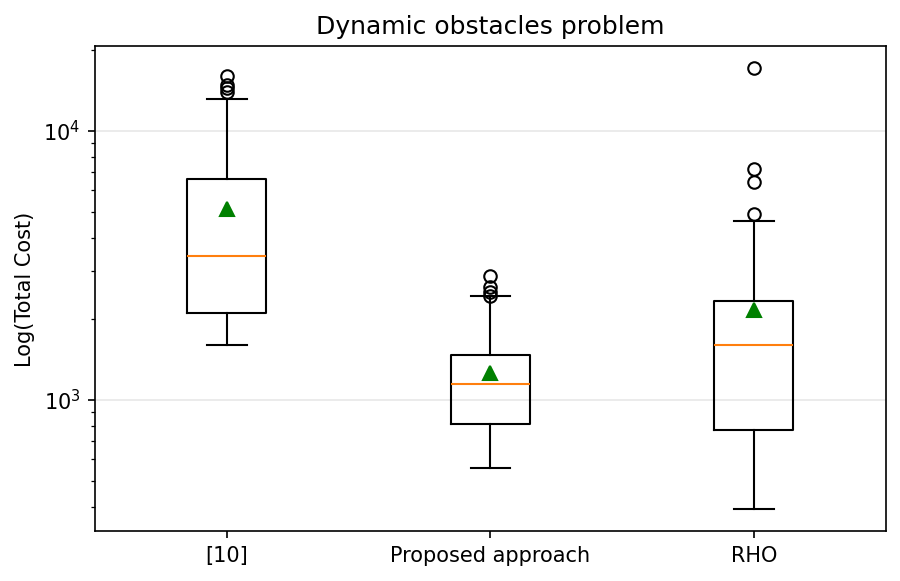}
 \caption{Dynamic obstacles problem -  Total cost (log scale) over 50 runs for the offline controller in~\cite{furieriPerformance}, a receding-horizon open-loop (RHO) planner, and the proposed approach.}
    \label{fig:boxplotTimevarying}
\end{figure}
\subsubsection{Computation of the upper bound of the $\ell_2$-gain}
\label{app:ell2gain}
For the pre-stabilized error dynamics, we consider the least favorable slope of the friction term
\[
G(v)=-b_1 v + b_2\tanh(v), 
\qquad 
\sigma_{\mathrm{wc}}:=-(b_1-b_2),
\]
which corresponds to the minimum-damping case since
\[
-b_1 \le \frac{\partial G}{\partial v} \le -(b_1-b_2).
\]
Using this worst-case slope, we obtain a linear error model of the form
\[
e_{t+1}=A_{\mathrm{wc}}e_t+B_u u_t+w_{t+1}.
\]
Let $\bar B:=[\,B_u\ \ I\,]$ and $\xi_t:=[\,u_t^\top\ \ w_{t+1}^\top\,]^\top$. 
We compute an upper bound on the $\ell_2$-gain of $\mathbfcal F_e:(\mathbf u,\mathbf w)\mapsto \mathbf e$ by finding $P\succ 0$, $\eta>0$, and $\rho>0$ such that
\begin{equation}\label{eq:lmi_gammaF}
\begin{bmatrix}
A_{\mathrm{wc}}^\top P A_{\mathrm{wc}}-P+\eta I & A_{\mathrm{wc}}^\top P \bar B\\
\bar B^\top P A_{\mathrm{wc}} & \bar B^\top P \bar B-\rho I
\end{bmatrix}\preceq 0.
\end{equation}
Then, with $V(e)=e^\top P e$, \eqref{eq:lmi_gammaF} implies
\[
V(e_{t+1})-V(e_t)\le -\eta \|e_t\|^2+\rho\bigl(\|u_t\|^2+\|w_{t+1}\|^2\bigr).
\]
Summing over time and using the convention $w_0=e_0$ yields
\[
\eta \|\mathbf e\|_2^2
\le
\lambda_{\max}(P)\|w_0\|^2+\rho\|\mathbf u\|_2^2+\rho\|w_{1:\infty}\|_2^2
\le
c\bigl(\|\mathbf w\|_2^2+\|\mathbf u\|_2^2\bigr),
\]
where $c:=\max\{\lambda_{\max}(P),\rho\}$. Therefore,
\[
\|\mathbf e\|_2
\le
\hat\gamma(\mathbfcal F_e)\bigl(\|\mathbf w\|_2+\|\mathbf u\|_2\bigr),
\qquad
\hat\gamma(\mathbfcal F_e):=
\max\!\left\{1,\sqrt{\frac{c}{\eta}}\right\}.
\]

\end{document}